\documentclass[11pt]{article}

\usepackage{amsmath,amssymb,amsfonts,graphicx,multicol,multirow, authblk}
\usepackage[small,it]{caption}
\usepackage{epsfig}
\usepackage{enumerate}
\usepackage{fullpage}
\usepackage[boxed,linesnumbered]{algorithm2e}
\usepackage[sanserif]{complexity}
\usepackage{graphics}
\usepackage{tikz}
\usepackage{amsthm}

\newclass{\TISP}{TISP}

\newcommand{\comment}[1]{}

\newcommand{\calS}{{\cal S}}
\newcommand{\calT}{{\cal T}}

\newcommand{\nll}{{\tt NULL}}
\newcommand{\no}{{\tt NO}}
\newcommand{\yes}{{\tt YES}}

\newlang{\LPGR}{LPGR}
\newlang{\LGGR}{LGGR}
\newlang{\oneLGGR}{1LGGR}
\newlang{\oneoneLGGR}{11LGGR}
\newclass{\ASC}{nSC}
\newclass{\PASC}{PASC}
\newclass{\ePASC}{\PASC_{\epsilon}}
\newclass{\eASC}{\ASC_{\epsilon}}

\SetKwInOut{Input}{Input}
\SetKwInOut{Output}{Output}
\SetKwInOut{Promise}{Promise}

\usepackage{epsfig}
\usepackage{url}

\newtheorem{theorem}{Theorem}
\newtheorem{lemma}[theorem]{Lemma}

\newtheorem{claim}[theorem]{Claim}
\newtheorem{proposition}[theorem]{Proposition}

\newtheorem{definition}{Definition}

\newcommand{\ignore}[1]{}

\newcommand{\newfontobj}[2]{
  \newcommand{#1}[1]{
    \expandafter\def\csname##1\endcsname{{#2 ##1}}}}

\comment{

\class{RL} \class{BPL}

\class{PSPACE}          
\class{AM}              
\class{AMTIME}
\class{NC}
\class{AMSUBEXP}
\class{MA}
\class{NP}
\class{NLIN}
\class{L}
\class{BPSPACE}
\class{LINSPACE}
\class{P}
\class{RP}
\class{BPP}
\class{RNC}
\class{PrAM}
\class{NPSV}
\class{advBPP}
\class{DTIME}
\class{E}
\class{ZPP}
\class{EXPSPACE}
\class{SUBEXP}
\class{PH}
\class{IP}
\class{SNP}
\class{SNQP}
\class{AVP}
\class{QP}
\class{SNTIME}
\class{DSPACE}
\class{BPSPACE}
\class{BPTIME}

\class{DTIME}

\class{AvgTIME} \class{DistNP} \class{AZPP} \class{AVE}
\class{DistE} \class{DistEXP} \class{E} \class{NE} \class{UE}
\class{EXP} \class{pcomp} \class{NEXP} \class{UP} \class{DTIME}
\class{NTIME} \class{R} \class{NSUBEXP} \class{SUBEXP} \class{PF}
\class{LINcomp}
\class{BPTIME}
\class{DistNLIN}
\class{AME}
\class{SIZE}
\class{NSIZE}
\class{SAT}
\class{SIZE}

}

\title{An $O(n^{\epsilon})$ Space and Polynomial Time Algorithm for Reachability in Directed Layered Planar Graphs}

\author{Diptarka Chakraborty\thanks{diptarka@cse.iitk.ac.in}}
\author{Raghunath Tewari \thanks{rtewari@cse.iitk.ac.in}}
\affil{Department of Computer Science and Engineering,\\ Indian
  Institute of Technology Kanpur,\\Kanpur, India} 

\begin{document}

\maketitle

\begin{abstract}
Given a graph $G$ and two vertices $s$ and $t$ in it, {\em graph reachability} is the problem of checking whether there exists a path from $s$ to $t$ in $G$. We show that reachability in directed layered planar graphs can be decided in polynomial time and $O(n^\epsilon)$ space, for any $\epsilon > 0$. The previous best known space bound for this problem with polynomial time was approximately $O(\sqrt{n})$ space \cite{INPVW13}.

Deciding graph reachability in {\SC} is an important open question in complexity theory and in this paper we make progress towards resolving this question.
\end{abstract}

\section{Introduction}

Given a graph and two vertices $s$ and $t$ in it, the problem of determining whether there is a path from $s$ to $t$ in the graph is known as the graph reachability problem. Graph reachability problem is an important question in complexity theory. Particularly in the domain of space bounded computations, the reachability problem in various classes of graphs characterize the complexity of different complexity classes. The reachability problem in directed and undirected graphs, is complete for the classes non-deterministic log-space (\NL) and deterministic log-space (\L) respectively \cite{LP82, Reingold08}. The latter follows due to a famous result by Reingold who showed that undirected reachability is in {\L} \cite{Reingold08}. Various other restrictions of reachability has been studied in the context of understanding the complexity of other space bounded classes (see \cite{RTV06, CRV07, rulcomplete}). Wigderson gave a fairly comprehensive survey that discusses the complexity of reachability in various computational models \cite{widgerson-survey}. 

The time complexity of directed reachability is fairly well understood. Standard graph traversal algorithms such as DFS and BFS solve this problem in linear time. We also have a $O(\log^2 n)$ space algorithm due to Savitch \cite{Sav70}, however it requires $O(n^{\log n})$ time. The question, whether there exists a single algorithm that decides reachability in polynomial time and polylogarithmic space is unresolved. In his survey, Wigderson asked whether it is possible to design a polynomial time algorithm that uses only $O(n^{\epsilon})$ space, for some constant $\epsilon < 1$ \cite{widgerson-survey}. This question is also still open. In 1992, Barnes, Buss, Ruzzo and Schieber made some progress on this problem and gave an algorithm for directed reachability that requires polynomial time and $O(n/2^{\sqrt{\log n}})$ space \cite{BBRS92}.

Planar graphs are a natural topological restriction of general graphs consisting of graphs that can be embedded on the surface of a plane such that no two edges cross. {\em Grid graphs} are a subclass of planar graphs, where the vertices are placed at the lattice points of a two dimensional grid and edges occur between a vertex and its immediate adjacent horizontal or vertical neighbor. 

Asano and Doerr provided a polynomial time algorithm to compute the \emph{shortest path} (hence can decide reachability) in grid graphs which uses $O(n^{1/2+\epsilon})$ space, for any small constant $\epsilon>0$ \cite{Asano11}. Imai et al extended this to give a similar bound for reachability in planar graphs \cite{INPVW13}. Their approach was to use a space efficient method to design a separator for the planar graph and use divide and conquer strategy. Note that although it is known that reachability in grid graphs reduces to planar reachability in logspace, however since this class (polynomial time and $O(n^{1/2+\epsilon})$ space) is not closed under logspace reductions, planar reachability does not follow from grid graph reachability. Subsequently the result of Imai et al was extended to the class of {\em high-genus} and {\em $H$-minor-free} graphs \cite{CPTVY14}. Recently Asano et al gave a $\tilde{O}(\sqrt{n})$ space and polynomial time algorithm for reachability in planar graphs, thus improving upon the previous space bound \cite{AKNW14}. More details on known results can be found in a recent survey article \cite{Vin14}.

In another line of work, Kannan et al gave a $O(n^{\epsilon})$ space and polynomial time algorithm for solving reachability problem in \emph{unique path graphs} \cite{KKR08}. Unique path graphs are a generalization of {\em strongly unambiguous} graphs and reachability problem in strongly unambiguous graphs is known to be in {\SC} (polynomial time and polylogarithmic space) \cite{BJLR91, logdcflinsc2}. Reachability in strongly unambiguous graphs can also be decided by a $O(\log^2 n/ \log \log n)$ space algorithm, however this algorithm requires super polynomial time \cite{AL98}. {\SC} also contains the class {\em randomized logspace} or {\RL} \cite{Nisan95}. We refer the readers to a recent survey by Allender \cite{allender-update} to further understand the results on the complexity of reachability problem in {\UL} and on certain special subclasses of directed graphs.

\subsection*{Our Contribution}
We show that reachability in directed layered planar graphs can be decided in polynomial time and $O(n^{\epsilon})$ space for any constant $\epsilon >0$.  A layered planar graph is a planar graph where the vertex set is partitioned into layers (say $L_0$ to $L_m$) and every edge occurs between layers $L_i$ and $L_{i+1}$ only. Our result significantly improves upon the previous space bound due to \cite{INPVW13} and \cite{AKNW14} for layered planar graphs.

\begin{theorem}
 \label{thm:layeredplanarreach}
For every $\epsilon >0$, there is a polynomial time and $O(n^{\epsilon})$ space algorithm that decides reachability in directed layered planar graphs.
\end{theorem}

Reachability in layered grid graphs is in {\UL} which is a subclass of {\NL} \cite{ABCDR09}. Subsequently this result was extended to the class of all planar graphs \cite{BTV}. Allender et al also gave some hardness results the reachability problem in certain subclasses of layered grid graphs. Specifically they showed that, {\oneLGGR} is hard for $\NC^1$ and {\oneoneLGGR} is hard for $\TC^0$ \cite{ABCDR09}. Both these problems are however known to be contained in {\L} though.

Firstly we argue that its enough to consider layered grid graphs (a subclass of general grid graphs). We divide a given layered grid graph into a courser grid structure along $k$ horizontal and $k$ vertical lines (see Figure \ref{fig:gridgraph}). We then design a modified DFS strategy that makes queries to the smaller graphs defined by these gridlines (we assume a solution in the smaller graphs by recursion) and visits every reachable vertex from a given start vertex. The modified DFS stores the highest visited vertex in each vertical line and the left most visited vertex in each horizontal line. We use this information to avoid visiting a vertex multiple number of times in our algorithm. We choose the number of horizontal and vertical lines to divide the graph appropriately to ensure that the algorithm runs in the required time and space bound. 

The rest of the paper is organized as follows. In Section \ref{sec:prelim}, we give some basic definitions and notations that we use in this paper. We also state certain earlier results that we use in this paper. In Section \ref{sec:LGGR}, we give a proof of Theorem \ref{thm:layeredplanarreach}.

\section{Preliminaries}
\label{sec:prelim}

We will use the standard notations of graphs without defining them explicitly and follow the standard model of computation to discuss the complexity measures of the stated algorithms. In particular, we consider the computational model in which an input appears on a read-only tape and the output is produced on a write-only tape and we only consider an internal read-write tape in the measure of space complexity. Throughout this paper, by $\log$ we mean logarithm to the base $2$. We denote the set $\{1,2,\cdots,n\}$ by $[n]$. Given a graph $G$, let $V(G)$ and $E(G)$ denote the set of vertices and the set of edges of $G$ respectively.
\begin{definition}[Layered Planar Graph]
 A planar graph $G=(V,E)$ is referred as \emph{layered planar} if it is possible to represent $V$ as a union of disjoint partitions, $V=V_1\cup V_2 \cup \dots \cup V_k$, for some $k>0$, and there is a planar embedding of edges between the vertices of any two consecutive partitions $V_i$ and $V_{i+1}$ and there is no edge between two vertices of non-consecutive partitions.  
\end{definition}
Now let us define the notion of layered grid graph and also note that grid graphs are by definition planar.
\begin{definition}[Layered Grid Graph]
 A directed graph $G$ is said to be a $n \times n$ \emph{grid graph} if it can be drawn on a square grid of size $n \times n$ and two vertices are neighbors if their $L_1$-distance is one. In a grid graph a edge can have four possible directions, i.e., north, south, east and west, but if we are allowed to have only two directions north and east, then we call it a \emph{layered grid graph}.
\end{definition}

We also use the following result of Allender et al to simplify our proof \cite{ABCDR09}.
\begin{proposition}[\cite{ABCDR09}]
 \label{prop:reductionLGGR}
 Reachability problem in directed layered planar graphs is log-space reducible to the reachability problem in layered grid graphs.
\end{proposition}

\subsection{Class {\ASC} and its properties}
$\TISP(t(n),s(n))$ denotes the class of languages decided by a deterministic Turing machine that runs in $O(t(n))$ time and $O(s(n))$ space. Then, $\SC = \TISP(n^{O(1)},(\log n)^{O(1)})$. Expanding the class {\SC}, we define the complexity class {\ASC} (short for {\tt near-SC}) in the following definition.

\begin{definition}[Complexity Class {\tt near-SC} or {\ASC}]
For a fixed $\epsilon > 0$, we define ${\eASC}:= \TISP(n^{O(1)},n^{\epsilon})$. The complexity class {\ASC} is defined as 
\[{\ASC}:=\bigcap_{\epsilon > 0} {\eASC}.\]
\end{definition}

We next show that {\ASC} is closed under log-space reductions. This is an important property of the class {\ASC} and will be used to prove Theorem \ref{thm:layeredplanarreach}. Although the proof is quite standard, but for the sake of completeness we provide it here. 

\begin{theorem}
 \label{thm:closure}
 If $A \le _l B$ and $B \in {\ASC}$, then $A \in {\ASC}$.
\end{theorem}
\begin{proof}
 Let us consider that a log-space computable function $f$ be the reduction from $A$ to $B$. It is clear that for any $x \in A$ such that $|x|=n$, $|f(x)| \le n^c$, for some constant $c>0$. We can think that after applying the reduction, $f(x)$ appears in a separate write-once output tape and then we can solve $f(x)$, which is an instance of the language $B$ and now the input length is at most $n^c$. Now take any $\epsilon>0$ and consider $\epsilon' =\frac{\epsilon}{c}>0$. $B \in {\ASC}$ implies that $B \in {\ASC_{\epsilon'}}$ and as a consequence, $A \in {\eASC}$. This completes the proof.
\end{proof}

\section{Reachability in Layered Planar Graphs}
\label{sec:LGGR}
In this section we prove Theorem \ref{thm:layeredplanarreach}. We show that the reachability problem in layered grid graphs, (denoted as {\LGGR}) is in {\ASC} (Theorem \ref{thm:LGGR}). Then by applying Proposition \ref{prop:reductionLGGR} and Theorem \ref{thm:closure} we have the proof of Theorem \ref{thm:layeredplanarreach}.

\begin{theorem}
 \label{thm:LGGR}
 {\LGGR} $\in$ {\ASC}.
\end{theorem}

To establish Theorem \ref{thm:LGGR} we define an auxiliary graph in Section \ref{sec:H} and give the required algorithm in Section \ref{sec:algo}.

\subsection{The Auxiliary Graph $H$}
\label{sec:H}

Let $G$ be a $n \times n$ layered grid graph. We denote the vertices in $G$ as $(i,j)$, where $0\le i,j \le n$. Let $k$ be a parameter that determines the number of pieces in which we divide $G$. We will fix the value of $k$ later to optimize the time and space bounds. Assume without loss of generality that $k$ divides $n$. Given $G$ we construct an auxiliary graph $H$ as described below. 

Divide $G$ into $k^2$ many blocks of dimension $n/k \times n/k$. More formally, the vertex set of $H$ is
\[
V(H) := \{(i,j) \ |\  \textrm{either $i$ or $j$ is a non-negative multiple of $n/k$}.\}
\]
Note that $V(H) \subseteq V(G)$. We consider $k^2$ many {\em blocks} $G_1,G_2,\cdots,G_{k^2}$, where a vertex $(i,j) \in V(G_l)$ if and only if $ i' \frac{n}{k}  \le i \le (i'+1) \frac{n}{k} $ and $ j' \frac{n}{k} \le j \le (j'+1) \frac{n}{k} $, for some integer $i' \ge 0$ and $j' \ge 0$ and the vertices for which any of the four inequalities is strict will be referred as \emph{boundary vertices}. Moreover, we have $l= i' \cdot k + j' +1 $. $E(G_l)$ is the set of edges in $G$ induced by the vertex set $V(G_l)$.

For every $i \in [k+1]$, let $L_h(i)$ and $L_v(i)$ denote the set of vertices, $L_h(i)=\{(i',j') | j' = (i-1) \frac{n}{k}\}$ and $L_v(i)=\{(i',j') | i' = (i-1) \frac{n}{k}\}$. When it is clear from the context, we will also use $L_h(i)$ and $L_v(i)$ to refer to the corresponding gridline in $H$. Observe that $H$ has $k+1$ vertical gridlines and $k+1$ horizontal gridlines.

For every pair of vertices $u,v \in V(G_l)\cap V(H)$, for some $l$, add the edge $(u,v)$ to $E(H)$ if and only if there is a path from $u$ to $v$ in $G_l$, unless $u,v \in L_v(i)$ or $u,v \in L_h(i)$, for some $i$. Also for every pair of vertices $u,v \in V(G_l)$, for some $l$, such that $u=(i_1,j_1)$ and $v=(i_2,j_2)$, where $j_1=j'  \frac{n}{k}$ and $j_2=(j'+1)\frac{n}{k} $, for some $j'$ or $i_1=i' \frac{n}{k} $ and $i_2=(i'+1)\frac{n}{k} $, for some $i'$, we add edge between $u$ and $v$ in the set $E(H)$ if and only if there is a path  from $u$ to $v$ in $G_l$ and we call such vertices as \emph{corner vertices}.

Before proceeding further, let us introduce a few more notations that will be used later. For $j \in [k]$, let $L_h(i,j)$ denote the set of vertices in $L_h(i)$ in between $L_v(j)$ and $L_v(j+1)$. Similarly we also define $L_v(i,j)$ (see Figure \ref{fig:gridgraph}). For two vertices $x,y \in L_v(i)$, we say $x \prec y$ if $x$ is {\em below} $y$ in $L_v(i)$. For two vertices $x,y \in L_h(i)$, we say $x \prec y$ if $x$ is {\em right of} $y$ in $L_h(i)$.

\begin{center}
 \begin{figure}
\centering
\begin{tikzpicture}[scale=.8,shorten >=.25mm,>=latex]
 \tikzstyle gridlines=[color=black!20,very thin]
 \draw[color=black!20,very thin] (0,0) grid (9,9);

 \foreach \x in {0,...,9}
  \foreach \y in {0,...,9}
   {
     \draw[fill,color=black!50] (\x,\y) circle (0.45mm);
   }


 \draw[left] node at (0,3) {$L_h(2)$};
 \draw[left] node at (0,6) {$L_h(3)$};

 \node at (3,-0.6) {$L_v(2)$};
 \node at (6,-0.6) {$L_v(3)$};
 
  \draw[fill=white!70!gray] (0,0) rectangle (3,3);

  \foreach \y in {1,...,9}
  {
    \draw[-,>=latex,thick] (0,\y-1)--(0,\y);
    \draw[-,>=latex,thick] (3,\y-1)--(3,\y);
    \draw[-,>=latex,thick] (6,\y-1)--(6,\y);
    \draw[-,>=latex,thick] (9,\y-1)--(9,\y);
  }
  
  \foreach \x in {1,...,9}
  {
    \draw[-,>=latex,thick] (\x-1,0)--(\x,0);
    \draw[-,>=latex,thick] (\x-1,3)--(\x,3);
    \draw[-,>=latex,thick] (\x-1,6)--(\x,6);
    \draw[-,>=latex,thick] (\x-1,9)--(\x,9);
  }
  \node at (1.5,1.5) {$G_1$};
  \node at (4.5,1.5) {$G_2$};
  \node at (7.5,1.5) {$G_3$};
  \node at (1.5,4.5) {$G_4$};
  \node at (4.7,4.5) {$G_5$};
  \node at (7.5,4.5) {$G_6$};
  \node at (1.5,7.5) {$G_7$};
  \node at (4.5,7.5) {$G_8$};
  \node at (7.5,7.5) {$G_9$};
  \node at (3.7,4.5) {$L_v(2,2)$};
  \node at (4.5,3.5) {$L_h(2,2)$};
  
  \node at (-0.2,-0.2) {$s$};
  \node at (9.2,9.2) {$t$};
  
  \draw[->,dashed,thick] (0,0) -- (0,1);
  \draw[->,dashed,thick] (0,1)-- (1,1) -- (7,1);
  \draw[->,dashed,thick] (7,1)-- (7,5); 
  \draw[->,dashed,thick] (7,5)-- (8,5);
  \draw[->,dashed,thick] (8,5)-- (8,8) -- (8,9);
  \draw[->,dashed,thick] (8,9)-- (9,9);


\end{tikzpicture}
\caption{An example of layered grid graph and its decomposition into blocks}
   \label{fig:gridgraph}
\end{figure}
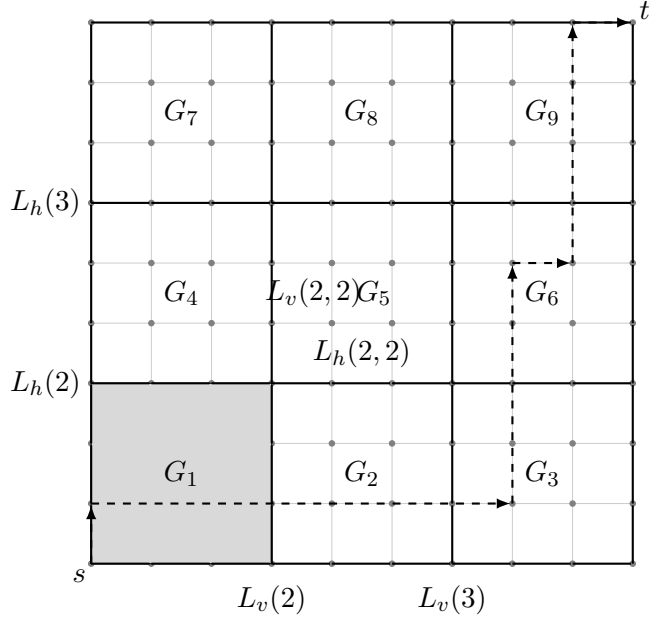
\end{center}

\begin{lemma}
\label{lem:auxiliary}
Let $u$ and $v$ be two boundary vertices in $G$ such that they belong to two different vertical or horizontal gridlines or $u\in G_i$ and $v \in G_j$, for $i \ne j$. There is a path from $u$ to $v$ in $G$ if and only if there is path from $u$ to $v$ in the auxiliary graph $H$.
\end{lemma}
\begin{proof}
 As every edge $(a,b)$ in $H$ corresponds to a path from $a$ to $b$ in $G$, so if-part is trivial to see. Now for the only-if-part, consider a path $P$ from $u$ to $v$ in $G$. $P$ can be decomposed as $P_1 P_2 \cdots P_r$, such that $P_i$ is a path from $x_i\in V(G_l)$ to $x_{i+1}\in V(G_l)$, where it re-enters $V(G_l)$ for the next time, for some $l$, and is of following two types:
 \begin{enumerate}
 \item $x_i$ and $x_{i+1}$ belong to different horizontal or vertical gridlines; or
 \item $x_i$ and $x_{i+1}$ are two corner vertices.
 \end{enumerate}

Now by the construction $H$, for every $i$, there must be an edge $(x_i,x_{i+1})$ in $H$ for both the above cases and hence there is a path from $u$ to $v$ in $H$ as well.
\end{proof}
Now we consider the case when $u$ and $v$ belong to the same vertical or horizontal gridlines.
\begin{claim}
\label{clm:gridreach}
Let $u$ and $v$ be two vertices contained in either $L_v(i)$ or $L_h(i)$ for some $i$. Then deciding reachability between $u$ and $v$ in $G$ can be done in logspace.
\end{claim}
\begin{proof}
Let us consider that $u,v\in L_v(i)$, for some $i$. As the graph $G$ under consideration is a layered grid graph, if there is a path between $u$ and $v$, then it must pass through all the vertices in $L_v(i)$ that lies in between $u$ and $v$. Hence just by exploring the path starting from $u$ through $L_v(i)$, we can check the reachability and it is easy to see that this can be done in logspace, because the only thing we need to remember is the current vertex in the path. Same argument will also work when $u,v\in L_h(i)$, for some $i$ and this completes the proof.
\end{proof}
Now we argue on the upper bound of the length of any path in the auxiliary graph $H$.

\begin{lemma}
\label{lem:pathlength}
For any two vertices $u,v \in V(H) $, any path between them is of length at most $2k+1$.
\end{lemma}
\begin{proof}
Consider any two vertices $u,v\in V(H)$ and a path $u=x_1 x_2\cdots x_r=v$, from $u$ to $v$, denoted as $P$. Now let us consider a bipartite (undirected) graph $K$, where $V(K)=A \cup B$ such that $A=\{x_i|i \in [r]\}$ and $B=\{L_v(i),L_h(j) | i,j \in [k+1]\}$. We add an edge $(a,b)$ in $E(K)$ if and only if $a \in b$, where $b=L_v(i)$ or $b=L_h(i)$, for some $i$. Now observe that by the construction of $H$, each $x_j$ belongs to different $L_v(i)$ or $L_h(i)$ unless $x_j$ is some corner vertex and in that case $x_j \in L_v(i),L_h(i')$, for some $i$ and $i'$. Moreover, if $x_j \in L_v(i)$ (or $L_v(i)$), but is not a corner vertex, then $x_{j+1}$ cannot be a corner vertex on $L_v(i)$ (or $L_h(i)$). As a consequence for every subset $S_A \subseteq A$, its neighbor set $N(S_A):=\{b \in B | \exists a \in S_A \text{ such that }(a,b) \in E(K) \}\subseteq B$ satisfies the condition that $|N(S_A)| \ge |S_A|$. Now we apply the Hall's Theorem \cite{LP86}, which states that a bipartite (undirected) graph $ G=(A \cup B,E) $ has a \emph{matching} if and only if for every $ S \subseteq A $, $ |N(S)| \ge |S| $. Hence there is a matching between $A$ and $B$ and as $|B|\le 2(k+1)$, so is $|A|$. This shows that the path $P$ is of length at most $2k+1$.
\end{proof}

\subsection{Description of the Algorithm}
\label{sec:algo}

We next give a modified version of DFS that starting at a given vertex, visits the set of vertices reachable from that vertex in the graph $H$. At every vertex, the traversal visits the set of outgoing edges from that vertex in an anticlockwise order.

In our algorithm we maintain two arrays of size $k+1$ each, say $A_v$ and $A_h$, one for vertical and the other for horizontal gridlines respectively. For every $i \in [k+1]$, $A_v(i)$ is the {\em topmost} visited vertex in $L_v(i)$ and analogously $A_h(i)$ is the {\em leftmost} visited vertex in $L_h(i)$. This choice is guided by the choice of traversal of our algorithm. More precisely, we cycle through the outgoing edges of a vertex in an anticlockwise order. 

We perform a standard DFS-like procedure, using the tape space to simulate a stack, say $S$. $S$ keeps track of the path taken to the current vertex from the starting vertex. By Lemma \ref{lem:pathlength}, the maximum length of a path in $H$ is at most $2k+1$. Whenever we visit a vertex in a vertical gridline (say $L_v(i)$), we check whether the vertex is lower than the $i$-th entry of $A_v$. If so, we return to the parent vertex and continue with its next child. Otherwise, we update the $i$-th entry of $A_v$ to be the current vertex and proceed forward. Similarly when visit a horizontal gridline (say $L_h(i)$), we check whether the current vertex is to the left of the $i$-th entry of $A_h$. If so, we return to the parent vertex and continue with its next child. Otherwise, we update the $i$-th entry of $A_h$ to be the current vertex and proceed. The reason for doing this is to avoid revisiting the subtree rooted at the node of an already visited vertex. The algorithm is formally defined in Algorithm \ref{algo:LGGR}.


\begin{algorithm}
\Input{The auxiliary graph $ H $, two vertices $u,v \in V(H)$}
\Output{{\yes} if there is a  path from $u$ to $v$; otherwise {\no}}
 
Initialize two arrays $A_v$ and $A_h$ and a stack $S$\;
   Initialize three variables $curr$, $prev$ and $next$ to {\nll}\;
   Push $u$ onto $S$\;
   \While{$S$ is not empty}{	
       $curr \leftarrow$ top element of $S$\;
       $next \leftarrow$ neighbor of $curr$ next to $prev$ in counter-clockwise order\;
		\While{$next \ne \nll$}{ \tcc{cycles through neighbors of curr}
			\If{$next=v$}{ 
        			\Return {\yes}\;
     	  	}
	       \If{$next \in L_v(i)$ for some $i$ and $A_v[i] \prec next$}{  
 	   		     $A_v[i] \leftarrow next $\;
                 {\bf break}\;
       }
       \If{$next \in L_h(i)$ for some $i$ and $A_h[i] \prec next$}{     
                 $A_h[i] \leftarrow next $\;
                 {\bf break}\;
       }
        $prev \leftarrow next$\;
        $next \leftarrow$ neighbor of $curr$ next to $prev$ in counter-clockwise order\; \tcc{{\nll} if no more neighbors are present}
        }

       \eIf{$next = \nll$}{
             remove $curr$ from $S$\;
             $prev \leftarrow curr$\;
       }
       {
		       add $next$ to $S$\;
             $prev \leftarrow \nll$\;
       }
   }
          
   \Return {\no}\;
	    
   \caption{{\tt AlgoLGGR}: Algorithm for Reachability in the Auxiliary Graph $H$}
   \label{algo:LGGR}
\end{algorithm}

%
%
%
%
%


\begin{lemma}
\label{lem:crossing}
Let $G_l$ be some block and let $x$ and $y$ be two vertices on the boundary of $G_l$ such that there is a path from $x$ to $y$ in $G$. Let $x'$ and $y'$ be two other boundary vertices in $G_l$ such that (i) there is a path from $x'$ to $y'$ in $G$ and (ii) $x'$ lies on one segment of the boundary of $G_l$ between vertices $x$ and $y$ and $y'$ lies on the other segment of the boundary. Then there is a path in $G$ from $x$ to $y'$ and from $x'$ to $y$. Hence, if $(x,y)$ and $(x',y')$ are present in $E(H)$ then so are $(x,y')$ and $(x',y)$.
\end{lemma}

\begin{proof}
Since $G$ is a layered grid graph hence the paths $x$ to $y$ and $x'$ to $y'$ must lie inside $G_l$. Also because of planarity, the paths must intersect at some vertex in $G_l$. Now using this point of intersection, we can easily show the existence of paths from $x$ to $y'$ and from $x'$ to $y$.
\end{proof}
 
 Lemma \ref{lem:reach} will prove the correctness of Algorithm \ref{algo:LGGR}.

 \begin{lemma}
 \label{lem:reach}
Let $u$ and $v$ be two vertices in $H$. Then starting at $u$ Algorithm \ref{algo:LGGR} visits $v$ if and only if $v$ is reachable from $u$.
 \end{lemma}
 \begin{proof}
It is easy to see that every vertex visited by the algorithm is reachable from $u$ since the algorithm proceeds along the edges of $H$. 

By induction on the shortest path length to a vertex, we will show that if a vertex is reachable from $u$ then the algorithm visits that vertex. Let $B_d(u)$ be the set of vertices reachable from $u$ that are at a distance $d$ from $u$. Assume that the algorithm visits every vertex in $B_{d-1}(u)$. Let $x$ be a vertex in $B_d(u)$. 
Without loss of generality assume that $x$ is in $L_v(i,j)$ for some $i$ and $j$. A similar argument can be given if $x$ belongs to a horizontal gridline. Further, let $x$ lie on the right boundary of a block $G_l$. Let $W_x = \{w \in B_{d-1}(u)| (w,x) \in E(H)\}$. Note that by the definition of $H$, all vertices in $W_x$ lie on the bottom boundary or on the left boundary of $G_l$.

Suppose the algorithm does not visit $x$. 
Since $x$ is reachable from $u$ via a path of length $d$, therefore $W_x$ is non empty. Let $w$ be the first vertex added to $W_x$ by the algorithm. Then $w$ is either in $L_h(j)$, or in $L_v(i-1)$. Without loss of generality assume $w$ is in $L_h(j)$. Let $z$ be the value in $A_v(i)$ at this stage of the algorithm (that is when $w$ is the current vertex). Since $x$ is not visited hence $x \prec z$. Also this implies that $z$ was visited by the algorithm at an earlier stage of the algorithm. Let $w'$ be the ancestor of $z$ in the DFS tree such that $w'$ is in $L_h(j)$. There must exist such a vertex because $z$ is above the $j$-th horizontal gridline, that is $L_h(j)$.

Suppose if $w'$ lies to the left of $w$ then by the description of the algorithm, $w$ is visited before $w'$. Hence $x$ is visited before $z$. On the other hand, suppose if $w'$ lies to the right of $w$. Clearly $w'$ cannot lie to the right of vertical gridline $L_v(i)$ since $z$ is reachable from $w'$ and $z $ is in $L_v(i)$. Let $w''$ be the vertex in $L_h(j+1)$ such that $w''$ lies in the tree path between $w'$ and $z$ (See Figure \ref{fig:reach}). Observe that all four vertices lie on the boundary of $G_l$. Now by applying Lemma \ref{lem:crossing} to the four vertices $w$, $x$, $w'$ and $w''$ we conclude that there exists a path from $w'$ to $x$ as well. Since $x \prec z$, $x$ must have been visited before $z$ from the vertex $w'$. In both cases, we see that   $z$ cannot be $A_v(i)$ when $w$ is the current vertex. Since $z$ was an arbitrary vertex such that $x \prec z$, the lemma follows.
\end{proof}

 \begin{center}
  \begin{figure}
\centering
\begin{tikzpicture}[scale=1,shorten >=.25mm,>=latex]



 \draw[-,>=latex,thick] (0,0)--(4,0);
 \draw[-,>=latex,thick] (0,0)--(0,4);
 \draw[-,>=latex,thick] (0,3)--(4,3);
 \draw[-,>=latex,thick] (3,0)--(3,4);
 \draw node at (4.6,0) {$L_h(j)$};
 \draw node at (4.8,3) {$L_h(j+1)$};

 \node at (0,-0.4) {$L_v(i-1)$};
 \node at (3,-0.4) {$L_v(i)$};

 \draw[->,dashed,thick] (1,0) -- (3,2);
  \node at (3.2,2.2) {$x$};
  \draw[->,dashed,thick] (2,0)-- (3,3.8);
  \node at (3.2,3.8) {$z$};
  \node at (1,-0.2) {$w$};
  \node at (2,-0.2) {$w'$};
  \node at (2.6,3.2) {$w''$};
  \node at (1.5,1.5) {$G_l$};


\end{tikzpicture}
\caption{Crossing between two paths}
   \label{fig:reach}
\end{figure}
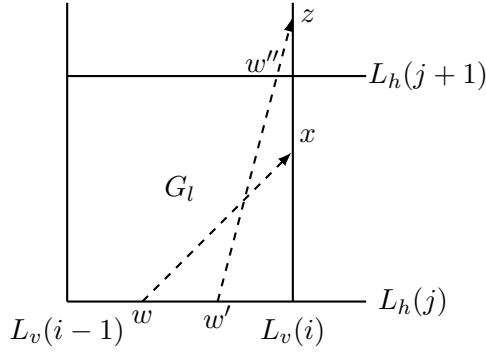
 \end{center}
 
 We next show Lemma \ref{lem:visitonce} which will help us to achieve a polynomial bound on the running time of Algorithm \ref{algo:LGGR}.
 
 \begin{lemma}
 \label{lem:visitonce}
 Every vertex in the graph $H$ is added to the set $S$ at most once in Algorithm \ref{algo:LGGR}.
 \end{lemma}
 \begin{proof}
Observe that a vertex $u$ in $L_v(i)$ is added to $S$ only if $A_v(i) \prec u$, and once $u$ is added, $A_v(i)$ is set to $u$. Also during subsequent stages of the algorithm, if $A_v(i)$ is set to $v$, then $u \prec v$. Hence $u \prec A_v(i)$. Therefore, $u$ cannot be added to $S$ again. 

We give a similar argument if $u$ is in $L_h(i)$. Suppose if $u$ is in $L_v(i)$ for some $i$ and $L_h(j)$ for some $j$, then we add $u$ only once to $S$. This check is done in Line 16 of Algorithm \ref{algo:LGGR}. However we update both $A_v(i)$ and $A_h(j)$.
 \end{proof}

Algorithm \ref{algo:LGGR} does not explicitly compute and store the graph $H$. Whenever it is queried with a pair of vertices to check if it forms an edge, it recursively runs a reachability query in the corresponding block and produces an answer. The base case is when a query is made to a grid of size $k \times k$. For the base case, we run a standard DFS procedure on the $k \times k$ size graph.

In every iteration of the {\em outer while} loop (Lines 4 -- 29) of Algorithm \ref{algo:LGGR}, either an element is added or an element is removed from $S$. Hence by Lemma \ref{lem:visitonce} the loop iterates at most $4nk$ times. The {\em inner while} loop (Lines 7 -- 21), cycles through all the neighbors of a vertex and hence iterates for at most $2n/k$ times. Each iteration of the inner while loop makes a constant number of calls to check the presence of an edge in a $n/k \times n/k$ sized grid. Let $\calT(n)$ and $\calS(n)$ be the time and space required to decide reachability in a layered grid graph of size $n \times n$ respectively. Then,
\[
\calT(n) = 
 \begin{cases}
 8n^2 (\calT(n/k) + O(1)) & \text{if } n > k\\
 O(k^2) & \text{otherwise}.
 \end{cases}
\]
Hence,
\[
\calT(n) = O\left(n^{3\frac{\log n}{\log k}}\right).
\]

Since we do not store any query made to the smaller grids, therefore the space required to check the presence of an edge in $H$ can be reused. $A_v$ and $A_h$ are arrays of size $k+1$ each. By Lemma \ref{lem:pathlength}, the number of elements in $S$ at any stage of the algorithm is bounded by $2k+1$. Therefore,
\[
\calS(n) = 
 \begin{cases}
 \calS(n/k) + O(k \log n) & \text{if } n > k\\
 O(k^2) & \text{otherwise}.
 \end{cases}
\]
Hence,
\[
\calS(n) = O\left(\frac{k}{\log k} \log^2 n+k^2\right).
\]

Now given any constant $\epsilon > 0$, if we set $k = n^{\epsilon/2}$, then we get $\calT(n) = O(n^{6/\epsilon})$ and $\calS(n) = O(n^{\epsilon})$. This proves Theorem \ref{thm:LGGR}.

\section*{Acknowledgement}
We thank N. V. Vinodchandran for his helpful suggestions and comments. The first author would like to acknowledge the support of Research I Foundation.

\bibliographystyle{alpha}
\bibliography{references}

\newcommand{\etalchar}[1]{$^{#1}$}
\begin{thebibliography}{AKNW14}

\bibitem[ABC{\etalchar{+}}09]{ABCDR09}
Eric Allender, David A.~Mix Barrington, Tanmoy Chakraborty, Samir Datta, and
  Sambuddha Roy.
\newblock Planar and grid graph reachability problems.
\newblock {\em Theory of Computing Systems}, 45(4):675--723, 2009.

\bibitem[AD11]{Asano11}
Tetsuo Asano and Benjamin Doerr.
\newblock Memory-constrained algorithms for shortest path problem.
\newblock In {\em CCCG}, 2011.

\bibitem[AKNW14]{AKNW14}
Tetsuo Asano, David Kirkpatrick, Kotaro Nakagawa, and Osamu Watanabe.
\newblock O(sqrt(n))-space and polynomial-time algorithm for the planar
  directed graph reachability problem.
\newblock Technical Report TR14-071, I, 2014.

\bibitem[AL98]{AL98}
Eric Allender and Klaus{-}J{\"{o}}rn Lange.
\newblock Ruspace(log \emph{n})
  {\textdollar}{\textbackslash}subseteq{\textdollar} {DSPACE}
  (log\({}^{\mbox{2}}\) \emph{n} / log log \emph{n}).
\newblock {\em Theory Comput. Syst.}, 31(5):539--550, 1998.

\bibitem[All07]{allender-update}
E.~Allender.
\newblock Reachability problems: An update.
\newblock {\em Computation and Logic in the Real World}, pages 25--27, 2007.

\bibitem[BBRS92]{BBRS92}
Greg Barnes, Jonathan~F. Buss, Walter~L. Ruzzo, and Baruch Schieber.
\newblock {A sublinear space, polynomial time algorithm for directed s-t
  connectivity}.
\newblock In {\em Structure in Complexity Theory Conference, 1992., Proceedings
  of the Seventh Annual}, pages 27--33, 1992.

\bibitem[BJLR91]{BJLR91}
Gerhard Buntrock, Birgit Jenner, Klaus-Jörn Lange, and Peter Rossmanith.
\newblock Unambiguity and fewness for logarithmic space.
\newblock In L.~Budach, editor, {\em Fundamentals of Computation Theory},
  volume 529 of {\em Lecture Notes in Computer Science}, pages 168--179.
  Springer Berlin Heidelberg, 1991.

\bibitem[BTV09]{BTV}
Chris Bourke, Raghunath Tewari, and N.~V. Vinodchandran.
\newblock Directed planar reachability is in unambiguous log-space.
\newblock {\em ACM Transactions on Computation Theory}, 1(1):1--17, 2009.

\bibitem[Coo79]{logdcflinsc2}
S.A. Cook.
\newblock {Deterministic CFL's are accepted simultaneously in polynomial time
  and log squared space}.
\newblock In {\em Proceedings of the eleventh annual ACM {S}ymposium on Theory
  of {C}omputing}, pages 338--345. ACM, 1979.

\bibitem[CPT{\etalchar{+}}14]{CPTVY14}
Diptarka Chakraborty, Aduri Pavan, Raghunath Tewari, N.~Variyam Vinodchandran,
  and Lin Yang.
\newblock New time-space upperbounds for directed reachability in high-genus
  and {\textdollar}h{\textdollar}-minor-free graphs.
\newblock {\em Electronic Colloquium on Computational Complexity {(ECCC)}},
  21:35, 2014.

\bibitem[CRV11]{CRV07}
Kai-Min Chung, Omer Reingold, and Salil Vadhan.
\newblock S-t connectivity on digraphs with a known stationary distribution.
\newblock {\em ACM Trans. Algorithms}, 7(3):30:1--30:21, July 2011.

\bibitem[INP{\etalchar{+}}13]{INPVW13}
T.~Imai, K.~Nakagawa, A.~Pavan, N.V. Vinodchandran, and O.~Watanabe.
\newblock An ${O}(n^{1/2+\epsilon}$)-{S}pace and {P}olynomial-{T}ime
  {A}lgorithm for {D}irected {P}lanar {R}eachability.
\newblock In {\em Computational Complexity (CCC), 2013 IEEE Conference on},
  pages 277--286, 2013.

\bibitem[KKR08]{KKR08}
Sampath Kannan, Sanjeev Khanna, and Sudeepa Roy.
\newblock {STCON in Directed Unique-Path Graphs}.
\newblock In Ramesh Hariharan, Madhavan Mukund, and V~Vinay, editors, {\em
  IARCS Annual Conference on Foundations of Software Technology and Theoretical
  Computer Science}, volume~2 of {\em Leibniz International Proceedings in
  Informatics (LIPIcs)}, pages 256--267, Dagstuhl, Germany, 2008. Schloss
  Dagstuhl--Leibniz-Zentrum fuer Informatik.

\bibitem[Lan97]{rulcomplete}
Klaus-J\"{o}rn Lange.
\newblock An unambiguous class possessing a complete set.
\newblock In {\em STACS '97: Proceedings of the 14th Annual Symposium on
  Theoretical Aspects of Computer Science}, pages 339--350, 1997.

\bibitem[LP82]{LP82}
Harry~R. Lewis and Christos~H. Papadimitriou.
\newblock Symmetric space-bounded computation.
\newblock {\em Theor. Comput. Sci.}, 19:161--187, 1982.

\bibitem[LP86]{LP86}
L.~Lovasz and M.D. Plummer.
\newblock Matching theory.
\newblock 1986.

\bibitem[Nis95]{Nisan95}
Noam Nisan.
\newblock {RL} $\subseteq$ {SC}.
\newblock In {\em In Proceedings of the Twenty Fourth Annual ACM Symposium on
  Theory of Computing}, pages 619--623, 1995.

\bibitem[Rei08]{Reingold08}
Omer Reingold.
\newblock Undirected connectivity in log-space.
\newblock {\em Journal of the ACM}, 55(4), 2008.

\bibitem[RTV06]{RTV06}
Omer Reingold, Luca Trevisan, and Salil Vadhan.
\newblock Pseudorandom walks on regular digraphs and the {RL} vs. {L} problem.
\newblock In {\em STOC '06: Proceedings of the thirty-eighth annual ACM
  Symposium on Theory of Computing}, pages 457--466, New York, NY, USA, 2006.
  ACM.

\bibitem[Sav70]{Sav70}
Walter~J. Savitch.
\newblock Relationships between nondeterministic and deterministic tape
  complexities.
\newblock {\em J. Comput. Syst. Sci.}, 4:177--192, 1970.

\bibitem[Vin14]{Vin14}
N.~V. Vinodchandran.
\newblock Space complexity of the directed reachability problem over
  surface-embedded graphs.
\newblock Technical Report TR14-008, I, 2014.

\bibitem[Wig92]{widgerson-survey}
Avi Wigderson.
\newblock The complexity of graph connectivity.
\newblock {\em Mathematical Foundations of Computer Science 1992}, pages
  112--132, 1992.

\end{thebibliography}



\end{document}